\documentclass[preprint,12pt,3p]{elsarticle}

\journal{arXiv}

\usepackage[dvipsnames]{xcolor}
\usepackage{amssymb,amsthm,amsmath,amsfonts,tikz}
\usepackage{subcaption}
\usepackage{enumerate}
\usepackage{tkz-graph}

\newtheorem{thm}{Theorem}[section]
\newtheorem{prop}[thm]{Proposition}
\newtheorem{cor}[thm]{Corollary}
\newtheorem{lem}[thm]{Lemma}

\newtheorem{defn}[thm]{Definition}

\DeclareMathOperator{\diam}{diam}
\DeclareMathOperator{\ndp}{ndp}
\DeclareMathOperator{\ddp}{dp}
\DeclareMathOperator{\DP}{DP}

\newcommand{\quot}[2]{#1/#2}

\begin{document}


\begin{frontmatter}

\title{Modular Decomposition of Graphs and the Distance Preserving Property}

\author[a,b]{Emad Zahedi}
\address[a]{Department of Mathematics, Michigan State University, \\
             East Lansing, MI 48824-1027, U.S.A.}
\address[b]{Department of Computer Science and Engineering,	Michigan State University\\
							East Lansing, MI 48824, U.S.A.}
\ead{Zahediem@msu.edu}

\author[c]{Jason P. Smith\fnref{fn1}}
\address[c]{Department of Computer and Information Sciences, University of Strathclyde, \\
             Glasgow, UK.}
\ead{jason.p.smith@strath.ac.uk}
\fntext[fn1]{J.P. Smith was supported by the EPSRC Grant EP/M027147/1.}

\begin{abstract}
Given a graph $G$, a subgraph $H$ is {\em isometric} if
$d_H(u,v) = d_G(u,v)$ for every pair $u,v\in V(H)$, where $d$ is the distance function. A graph $G$ is {\em distance preserving (dp)} if it has an isometric subgraph of every possible order. A graph is \emph{sequentially distance preserving (sdp)} if its vertices can be ordered such that deleting the first $i$ vertices results in an isometric subgraph, for all $i\ge1$. We introduce a generalisation of the lexicographic product of graphs, which can be used to non-trivially describe graphs. 
This generalisation is the inverse of the modular decomposition of graphs, which divides the graph into disjoint clusters called {\em modules}. Using these operations, we give a necessary and sufficient condition for graphs to be dp. Finally, we show that the Cartesian product of a dp graph and an sdp graph is dp.
\end{abstract}

\begin{keyword}
distance preserving, isometric, modular decomposition, lexicographic product, Cartesian product.
\end{keyword}

\end{frontmatter}

\section{Introduction}\label{section:I}
Many problems in graph theory can be tackled by decomposing a graph into smaller pieces and then studying the problem on these parts individually. There are many different ways to decompose a graph that have been applied to a variety of problems. In this paper we use modular decompositions of graphs to study the distance preserving property. Modular decomposition has been used to solve many problems, see~\cite{Gal67,Moh85,Moh84,Sid86}.

We call a subgraph {\em isometric} if the distance between any pair of vertices is the same as in the original graph. Distance properties and isometric subgraphs have been previously used in network clustering~\cite{NF2, nussbaum2013clustering}.
A graph is {\em distance preserving}, for which we use the abbreviation dp, if it has an isometric subgraph of every possible order. Distance preserving graphs have been studied in the literature, see \cite{esfahanian2014constructing,NF2,SmithZahedi16,zahedi2015distance}.

In~\cite{khalifeh2015distance} the distance preserving property is investigated  when taking products of graphs. Graph products are operations which take two graphs $G$ and $H$ and produce a graph with vertex set $V(G)\times V(H)$ and  certain conditions on the edge set, see~\cite{imrich2000product}. Two such products were considered in~\cite{khalifeh2015distance}, lexicographic product and Cartesian product. 
The purpose of this work is to generalise certain results from that article.
Various invariants of lexicographic products of graphs have been studied in the literature, see \cite{anand2012convex,vcivzek1994chromatic,yang2013connectivity}. The Cartesian product is a well-known graph product, in part because of Vizing's Conjecture~\cite{vizing1963cartesian}, and has been considered by many authors, such as~\cite{aurenhammer1992cartesian, caceres2007metric, khalifeh2009some, yousefi2008pi}.

The \emph{lexicographic product} $G[H]$ replaces every vertex of the graph $G$ with the graph $H$. We introduce the \emph{generalised lexicographic product} $G[\mathcal{H}]$ which replaces each vertex~$v$ of the graph $G$ with a graph $H_v\in\mathcal{H}$, where $\mathcal{H}$ is a set of graphs indexed by the vertices of~$G$. This can be viewed as a generalisation of the traditional lexicographic product  because setting~$H_v=H$, for all vertices $v$ of $G$ results in the lexicographic product~$G[H]$. Moreover, we see that any graph $M$ can be represented using the generalised lexicographic product, that is, $M$ is isomorphic to $G[\mathcal{H}]$ for some $G$ and $\mathcal{H}$.

The generalised lexicographic product has appeared in various forms in the literature. 
This operation is equivalent to applying a substitution, as first defined in \cite{Chv75}, to every vertex in the graph. One example of the implicit use of the generalised lexicographic product is Lov\'asz's proof of the perfect graph theorem \cite{Lov72} which uses the multiplication of vertices of a graph~$G$, which is equivalent to the generalised lexicographic product $G[\mathcal{H}]$ with~$H_v=\overline{K}_{h_v}$ for every vertex $v$ of $V(G)$, where $h_v\ge 1$ and $\overline{K}_{h_v}$ is the empty graph with~$h_v$ vertices.

A \emph{module} in a graph $M$ is an induced subgraph $H$ whose vertices share the same neighbourhood outside of $H$. 
A \emph{modular decomposition} of a graph $M$ is a collection of modules of $M$, where every vertex of $M$ appears in exactly one module. The neighborhood condition forces empty or complete bipartite graphs between modules. 
There are various polynomial time algorithms for computing the modular decomposition of a graph, see \cite{habib2010survey}. Given a modular decomposition $\mathcal{H}$ of $M$ we define the \emph{quotient graph} of $M$ with respect to~$\mathcal{H}$, denoted~$\quot{M}{\mathcal{H}}$, as the graph obtained by mapping each module of $\mathcal{H}$ to a single vertex, where there is an edge between two vertices of $\quot{M}{\mathcal{H}}$ if and only if there are edges between the vertices of the corresponding modules in $M$. The generalised lexicographic product can be consider as the inverse of the modular decomposition operation, thus $M$ is isomorphic to~$(\quot{M}{\mathcal{H}})[\mathcal{H}]$.

A \emph{split decomposition} of a graph is a modular decomposition into two modules both with order greater than $1$. 
Split decompositions have been used to study distance hereditary graphs, that is, graphs in which every induced subgraph is isometric. It was shown in \cite{bandelt1986distance} that the distance hereditary property is equivalent to a graph being totally decomposable using  split decompositions. See~\cite{GP08} for a definition of totally decomposable and a general overview of split decompositions.

In Section~\ref{graph_nature} we formally define the generalised lexicographic product and modular decomposition. We present a result that a quotient of a graph is minimal if and only if its corresponding modules are maximal, provided the quotient has at least three vertices. 
This strengthens some of the existing results in this area, see \cite{habib2010survey}. In Section~\ref{dp_graph} a necessary and sufficient condition is given for
graphs of the form $G[\mathcal{H}]$  to be dp. This condition implies that if $G$ is dp then $G[\mathcal{H}]$ is dp. Moreover, all isometric subgraphs of 
$G[\mathcal{H}]$ are characterized in this section.

In Section~\ref{section:F} we consider the Cartesian product of graphs. 
A graph $G$ is \emph{sequentially distance preserving}, which we abbreviate to sdp, if we can order the vertices $v_1,\ldots,v_n$ of~$G$ 
such that deleting the first $i$ vertices results in an isometric subgraph, for all $1\le i\le n$. 
In~\cite{khalifeh2015distance} it was shown that the Cartesian product $G\,\Box\,H$ of two graphs $G$ and $H$ is sdp if and only if $G$ and $H$ are sdp. Furthermore, it was conjectured that if $G$ and $H$ are dp then~$G\,\Box\,H$ is dp. We prove an intermediate result, namely that if $G$ is sdp and $H$ is dp, then~$G\,\Box\,H$ is dp.

\section{Generalised Lexicographic Product and Modular Decomposition}\label{graph_nature}
In this framework we assume all graphs are finite, nonempty, simple and connected, unless otherwise stated. We refer the reader to \cite{bondy1976graph} for a general overview of graph theory, which includes any definitions and notation not given in this paper. We let $|G|$ be the number of vertices of $G$ and denote the vertices and edges of a graph $G$ by $V(G)$ and $E(G)$, respectively. Given two graphs $G$ and $H$, let~$G-H$ be the graph induced by~$V(G)\setminus V(H)$.

In this section we introduce two graph operations, the generalised lexicographic product and modular decomposition. We note that these two operations are the inverses of each other. First we introduce the generalised lexicographic product. Recall that the {\em lexicographic product} $G[H]$ of graphs $G$ and $H$ is the graph with vertex set $V(G)\times V(H)$ and edge set
$$E(G[ H]) =\{(a,x)(b,y)\ |\text{ $ab\in E(G)$, or $xy\in E(H)$ and $a=b$}\}.$$	The reader can consult the book of Imrich and Klavzar~\cite{imrich2000product}, for more details about graph products.

\begin{defn}\label{facial}
Let $G$ be a graph and $\mathcal{H} = \{H_v\}_{v\in V(G)}$ be a set of graphs. 
Define the \emph{generalised lexicographic product} $G[\mathcal{H}]$ as the graph with vertex and edge sets
\begin{align*} 
V(G[\mathcal{H}]) = \bigcup_{v\in V(G)}\left( \{v\}\times V(H_v)\right),
\end{align*}
\begin{align*}
 E(G[\mathcal{H}])  = &\{ (u,x)(v,y)\ |\ uv\in E(G)\} \cup\bigcup_{v\in V(G)}\{(v,x)(v,y)\ |\ xy\in E(H_v)\}.
\end{align*}
\end{defn}

In other words $G[\mathcal{H}]$ is constructed by replacing every vertex $v\in V(G)$ with the graph~$H_v$, and the edges between $H_u$ and $H_v$ form a complete bipartite graph or the empty graph depending on whether $uv\in E(G)$ or $uv\not\in E(G)$, respectively. 
To clarify the notation Figure~\ref{fig:GFH} is given as an example. Note that if $H_v=H$, for all $v\in V(G)$, then $G[\mathcal{H}]$ is the lexicographic product graph $G[H]$.

\begin{figure}
    \centering
    \includegraphics{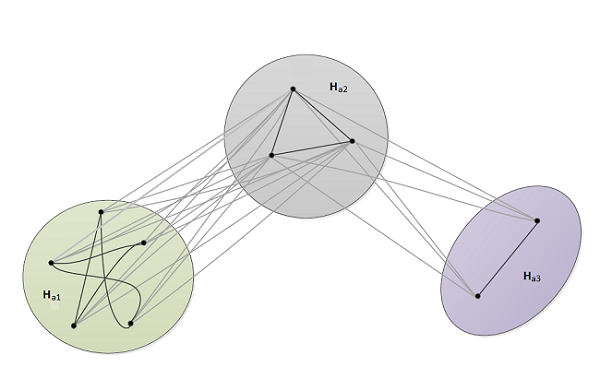}
		\caption{A graph $G[\{H_{a_1},H_{a_2}, H_{a_3}\}]$,
		where $G$ is the $2$-path $a_1a_2a_3$ and $H_{a_1}=C_5$, $H_{a_2}=K_3$ and $H_{a_3}=K_2$}  
\label{fig:GFH}
\end{figure}

The inverse of this operation has been well studied and is known as the modular decomposition of a graph, see \cite{habib2010survey} for an overview. The \emph{neighbourhood} of a vertex $v\in V(G)$, denoted~$N_G(v)$, is the set of all vertices in $G$ joined by an edge to $v$. Moreover, given a subgraph $A$ of $G$ let $N_G(A)=\cup_{v\in V(A)}N_G(v)\setminus V(A)$.

\begin{defn}
Let $H$ be a subgraph of a graph $M$. We call $H$ a \emph{module} of $M$ if~$N_M(u)\setminus V(H)=N_M(H)$, for all $u\in V(H)$. 
The module $H$ is \emph{maximal} if there is no module $H'$ of~$M$ such that $H\subsetneqq H'\subsetneqq M$. 
 A module of $M$ is \emph{trivial} if it is a single vertex or the whole graph. 
A \emph{modular partition} $\mathcal{H}$ of~$M$ is a set of disjoint modules of~$M$ such that $V(M)=\bigcup_{H\in\mathcal{H}}V(H)$. 
Two modules $H$ and $H'$ of a partition are said to be \emph{adjacent} if $(u,v)\in E(M)$ for every $(u,v)\in V(H)\times V(H')$.
A trivial or maximal decomposition of a graph is the modular decomposition where every module is trivial or maximal, respectively.
\end{defn}

For example, $H_{a_1}\cup H_{a_3}$ and $H_{a_2}$ are some of the modules in Figure~\ref{fig:GFH}. Moreover, deleting any vertex from $H_{a_2}$ gives a maximal module and the modules $H_{a_1}$ and $H_{a_2}$ are adjacent, but the modules $H_{a_1}$ and $H_{a_3}$ are not adjacent.

\begin{defn}
Let $M$ be a graph with a modular partition $\mathcal{H}$. The \emph{quotient graph} $\quot{M}{\mathcal{H}}$ is the graph with a single vertex $v_H$ for each $H\in\mathcal{H}$ and an edge between $v_{H}$ and~$v_{H'}$ if and only if $H$ and $H'$ are adjacent in $M$. 
We say that $\quot{M}{\mathcal{H}}$ is a \emph{minimal quotient graph} of $M$ if $\quot{M}{\mathcal{H}}$ contains no non-trivial modules.
\end{defn}

Note that the quotient operation and generalised lexicographic product are inverses of each other up to isomorphism, that is, $M\cong(\quot{M}{\mathcal{H}})[\mathcal{H}]$ and $G\cong\quot{(G[\mathcal{H}])}{\mathcal{H}}$, where $\cong$ denotes that two graphs are isomorphic.
We say that a graph $M$ can be \emph{represented} by a graph $G$ and set $\mathcal{H}$ if $M\cong G[\mathcal{H}]$. For example in Figure~\ref{fig:GFH}, the graph can be represented as~$K_2[\{H_{a_1}\cup H_{a_3}, H_{a_2}\}]$.

Next we present a useful lemma on the union of modules and then present the main result of this section. 

\begin{lem}\label{lem:union}
If $H$ and $K$ are both modules of $M$ with $V(H)\cap V(K)\not=\emptyset$, then $H\cup K$ is also a module of $M$.
\begin{proof}
Any vertex $b\in M- (H\cup K)$ is either a neighbour of all or none of $V(H)$. If $b$ is a neighbour of all of $V(H)$, then it is a neighbour of all of $V(H)\cap V(K)$, so it is a neighbour of all of $V(K)$. Similarly if $b$ is a neighbour of none of $V(H)$, then it is a neighbour of none of $V(K)$. Therefore, every element of $V(H)\cup V(K)$ has the same neighbours in~$M- (H\cup K)$.
\end{proof}
\end{lem}

\begin{thm}\label{MinMax}
Consider a graph $G$ with at least three vertices. The graph $G$ is a minimal quotient graph of 
$M=G[\mathcal{H}]$ if and only if $\mathcal{H}$ is a maximal modular decomposition.
\end{thm}
\begin{proof}
By definition a graph $G$ is a non-minimal quotient graph of $M$ if and only if $G$ contains a non-trivial module $K$. First we consider the forward direction, if $G$ contains a non-trivial module $K$, then $H_v$ is a non-maximal module in $M$ for any $v\in K$. To see the backwards direction suppose $H_v$ is a non-maximal module in $M$, so there is a maximal module $H'\supset H_v$. Furthermore, there is some other module $H_u$ with $H'\cap H_u\not=\emptyset$, so $H'\cup H_u$ is also a module by Lemma~\ref{lem:union}. However, as $H'$ is a maximal module we must have $H'\cup H_u=M$. If $H_u$ is the only other modules in $\mathcal{H}$ then $G$ has only two vertices. If there are $k>2$ modules then $\mathcal{H}$ contains $k-1>1$ modules and as these modules are all contained in one larger module the corresponding vertices in $G$ must form a module in $G$, so $G$ is not minimal.
\end{proof}
In the proof of Theorem~\ref{MinMax} the requirement that $G$ has at least three vertices is only needed for the backwards direction, so we get the following corollary:
\begin{cor}
If $\mathcal{H}$ is a maximal modular decomposition of $M$, then $\quot{M}{\mathcal{H}}$ is a minimal quotient graph.
\end{cor}
However it is necessary that the graph has at least three vertices for the forwards direction. To see this consider $K_4$ and the modular decomposition $\mathcal{H}$ partitioning $K_4$ into two modules of three vertices and one vertex. This is not a maximal modular partition but~$\quot{K_4}{\mathcal{H}}$ equals $K_2$ which is the minimal quotient graph of $K_4$.

Theorem~\ref{MinMax} is similar to Theorem 2 in \cite{habib2010survey}, but gives an equivalence statement rather than just a necessary condition. Moreover, the condition in Theorem 2 of \cite{habib2010survey} states that the graph $M=G[\mathcal{H}]$ must have a connected complement graph. The complement graph $\bar{M}$ of $M$ is the graph with the same vertices as $M$ and $xy$ is an edge in $\bar{M}$ if and only if $xy$ is not an edge in $M$. In fact the condition that $M=G[\mathcal{H}]$ must have a connected complement is equivalent to our condition that $|G|\ge3$, as can be seen by the following result:

\begin{lem}
The graph $M$ has a disconnected complement graph if and only if $K_2$ is a quotient graph of $M$.
\begin{proof}
The graph $\bar{M}$ is disconnected with components $A$ and $B$ if and only if there is no edge in $\bar{M}$ between any vertices $a\in V(A)$ and $b\in V(B)$, which is equivalent to the graphs induced by $V(A)$ and $V(B)$ in $M$ being a modular decomposition $\mathcal{H}$ with $\quot{M}{\mathcal{H}}=K_2$.
\end{proof}
\end{lem}

We can also determine when a graph has a unique maximal modular decomposition and unique minimal quotient graph.

\begin{lem}\label{lem:uniqueMax}
If $K_2$ is not a quotient graph of $M$, then $M$ has a unique maximal modular decomposition.
\begin{proof}
Suppose that $M$ has two different maximal modular decompositions $\mathcal{H}$ and $\mathcal{H}'$. There must exist a set $A$ that is the nonempty intersection of a pair $H\in\mathcal{H}$ and $H'\in\mathcal{H}'$ with $H\not=H'$. As both $H$ and $H'$ are modules $H\cup H'$ is also a module by Lemma~\ref{lem:union}. So the only way that $H$ and $H'$ are maximal is if $H\cup H'=M$. However $M- H$ is also a module. To see this first note that every element of $H'$ has the same neighbours in $H$, because there is at least one element $h\in H$ with $h\not\in H'$, so $h$ is either a neighbour of all or none of $H'$. Moreover, since every element of $H$ has the same neighbours either all elements of $H$ are neighbours of all elements of $H'$ or none are. Furthermore, $M-H\subseteq H'$ so every element of $M-H$ has the same neighbours in $H$, thus $M-H$ is a module. Therefore, $H$ and $M- H$ form a modular decomposition of $M$, so $K_2$ is a quotient graph of $M$.
\end{proof}
\end{lem}

\begin{cor}
Every graph $M$ has a unique minimal quotient graph.
\begin{proof}
If $K_2$ is a quotient graph of $M$, then $K_2$ is the unique minimal quotient graph. Otherwise, $M$ has a unique maximal modular decomposition $\mathcal{H}$ by Lemma~\ref{lem:uniqueMax}, so $\quot{M}{\mathcal{H}}$ is the unique minimal quotient graph.
\end{proof}
\end{cor}

Note that if a graph $G$ has a modular decomposition with the quotient graph $K_2$, where both modules are non-trivial, then this is a split decomposition.


\section{Distance Preserving Graphs}\label{dp_graph}
In this section we investigate some conditions under which $G[\mathcal{H}]$ is distance preserving. A \emph{path} in a graph is a sequence of distinct vertices with an edge between every consecutive pair. The {\em distance} between vertices $u,v$ in $G$, denoted $d_G(u,v)$, is the minimal length of a path connecting these vertices.
 If it is clear from  context we use $d(u,v)$, instead of $d_G(u,v)$. A path $\rho$ from $u$ to $v$ with length $d(u,v)$ is called $u$--$v$ {\em geodesic}. An induced subgraph $H$ of a graph $G$ is called an {\em isometric} subgraph, denoted $ H \leq G $, if $d_H(u,v)=d_G(u,v)$ for every pair of vertices $u,v \in V(H)$. A graph $G$ is called {\em distance preserving (dp)} if it has an~$i$-vertex isometric subgraph for every $1 \leq i \leq |V(G)|$.
 
 We begin by considering the relationship between the geodesic paths in $G$ and the geodesic paths in $G[\mathcal{H}]$:

\begin{lem}\label{lem:geodesics}
Consider a path $g_1,g_2,\ldots,g_\ell$ in $G$. A path $(g_1,h_1),(g_2,h_2),\ldots,(g_\ell,h_\ell)$ is $(g_1,h_1)$--$(g_\ell,h_\ell)$ geodesic in $G[\mathcal{H}]$ if and only if~$g_1,g_2,\ldots,g_\ell$ is $g_1$--$g_\ell$ geodesic in $G$.
\begin{proof}
First consider the forward direction, so suppose $(g_1,h_1),\ldots,(g_\ell,h_\ell)$ is $(g_1,h_1)$--$(g_\ell,h_\ell)$ geodesic. If $g_1,\ldots,g_\ell$ is not $g_1$--$g_\ell$ geodesic, then there is a~$g_1$--$g_\ell$ geodesic path $g'_1,\ldots,g'_k$, where $g'_1=g_1$, $g'_k=g_\ell$ and $k<\ell$. However, this would imply that there is a $(g_1,h_1)$--$(g_\ell,h_\ell)$ geodesic path $(g'_1,h'_1),\ldots,(g'_k,h'_k)$, where $h'_1=h_1$, $h'_k=h'_\ell$ and~$h'_i\in V(H_{g'_i})$ for all $1<i<k$, which contradicts $(g_1,h_1),\ldots,(g_\ell,h_\ell)$ being~$(g_1,h_1)$--$(g_\ell,h_\ell)$ geodesic. The backwards direction follows by an analogous argument.
\end{proof}
\end{lem}

\begin{lem}\label{L1}
Consider a connected graph $G$ with $|G| \ge 2$ and a set of graphs  $\mathcal{H} = \{H_v\}_{v\in V(G)}$.
\begin{enumerate}[(a)]
\item If $x\in V(H_u)$ and $y\in V(H_v)$ are distinct vertices, with $u,v\in V(G)$, then:
$$
d_{G[\mathcal{H}]}((u,x),(v,y)) =
\begin{cases}
d_G(u,v), &\mbox{if } u\ne v, \\
2, &\mbox{if } u=v \textnormal{ and } xy\notin E(H_u), \\
1, &\text{otherwise}.
\end{cases}
$$\label{casea}
\item If $u,v$ are distinct vertices of $G$, then $d_G(u,v)= d_{G[\mathcal{H}]}((u,x),(v,y))$, for any $x\in H_u$ and $y\in H_v$.\label{caseb}
\end{enumerate}
\end{lem}
\begin{proof}
First consider part \eqref{casea}. If $u\not=v$, the result follows by Lemma~\ref{lem:geodesics}. If $u=v$ and~$xy\notin E(H_u)$, then we have an $x$--$y$ geodesic path $xzy$, where $z\in V(H_w)$ and $w$ is any neighbour of $u$ in $G$. Finally, if $u=v \textnormal{ and } xy\in E(H_u)$, then $(u,x)(v,y)\in E(G[\mathcal{H}])$. This completes part~\eqref{casea}. Part~\eqref{caseb} follows by Lemma~\ref{lem:geodesics}.
\end{proof}

\begin{cor}\label{cor:connected}
The graph $G[\mathcal{H}]$ is connected if and only if $G$ is connected.
\begin{proof}
A graph is connected if and only if the distance between all vertices is finite. Therefore, the result follows immediately from Lemma~\ref{lem:geodesics}.
\end{proof}
\end{cor}

Note that Lemma~\ref{L1} and Corollary~\ref{cor:connected} generalise Lemma 3.1 in \cite{khalifeh2015distance} from the lexicographic product to the generalised lexicographic product. In the remainder of this section we generalise some more results of Section 3 in \cite{khalifeh2015distance}.

In order to state the main theorem of this section we need some notation. Let 
$$\ndp(G) = \{k\ | \text{ G has no isometric subgraph with k vertices}\},$$
so a graph $G$ is dp if and only if $\ndp(G) = \emptyset$. If $a$ and $b$ are integers with $a < b$, then let~$[a, b] = \{a, a + 1, a + 2, \ldots , b\}$. Given a subgraph $M$ of $G[\mathcal{H}]$, let $\pi(M)$ be the induced subgraph of $G$ with the vertex set:
 $$ V(\pi(M)) = \{a\in V(G) \ |\ (a,x)\in M\text{ for some }x\in H_a\}.$$

\begin{thm}\label{thm2}
Let $G$ be a connected graph with $|G| \ge 2$. Any generalised lexicographic product graph  
$G[\mathcal{H}]$ is $dp$ if and only if 
for any $k \in \ndp(G)$, there is a subgraph $L \leq G$ with $|L|< k \le \sum_{u\in L}|H_u|$. 
\end{thm}

\begin{proof}
We claim that for an induced subgraph $M$ of $G[\mathcal{H}]$, with $\pi(M)$ having at least two vertices,
\begin{align}\label{project}
\text{ $M \leq G[\mathcal{H}]$ if and only if  $\pi(M) \leq G.$  } 
\end{align}

To prove the backwards direction of the claim, assume that $\pi(M) \le G$ and consider distinct
vertices $(u,x),(v,y) \in V (M)$.
If $u\neq v$, then note that~$\pi(M)$ can be considered as a quotient graph of $M$, so Lemma~\ref{L1}\eqref{casea} and $\pi(M)\le G$ gives
$$ d_{M}((u,x),(v,y)) = d_{\pi(M)}(u, v) = d_G(u, v) = d_{G[\mathcal{H}]}((u,x),(v,y)).$$ If $u = v$ and $xy\notin E(H_u)$, then a similar proof shows that $ d_{M}((u,x),(v,y)) = 2 = d_{G[\mathcal{H}]}((u,x),(v,y)).$
 Finally, if $u = v$ and $xy\in E(H_u)$, then since $M$ is induced we have 
 $ d_{M}((u,x),(v,y)) = 1 = d_{G[\mathcal{H}]}((u,x),(v,y)).$  The forward direction of the claim follows by an analogous argument.

Now we prove the theorem. By the definition of the quotient graph we know that $|\pi(M)| \le |M| \le \sum_{u\in \pi(M)}|H_u|$. Statement~\eqref{project} implies that $G[\mathcal{H}]$ is dp if and only if 
\begin{align}\label{eq2} 
 \bigcup_{L\le G}\big[|L|, \sum_{u\in L}|H_u|\big] = \big[1, \sum_{u\in G}|H_u|\big].
\end{align}
Since $1, 2, |G|$ are never in $\ndp(G)$, Equality \eqref{eq2} is equivalent to: 
for any $k\in \ndp(G)$, there is an $L \leq G$ with $|L|< k \le \sum_{u\in L}|H_u|$.
\end{proof}

Theorem~\ref{thm2} generalises Theorem 3.2 in~\cite{khalifeh2015distance}. 
Suppose $G$ has isometric subgraphs with $a$ and $b$ vertices. Then we say two elements  $a,b$ {\em bound a non-dp interval} if the set of integers $c$ with $a< c < b$ is nonempty and consists only of elements in~$\ndp(G)$.

\begin{cor}\cite[Theorem 3.2]{khalifeh2015distance}
Let $G$ be a connected graph with $|G|\ge2$ and $H$ be an arbitrary graph with $|H|=n$. Then 
$G[H] \text{\  is dp if and only if \  } b\leq an+1$ 
for every pair~$a,b \in \ndp(G)$ bounding a non-dp interval.
\end{cor}
\begin{proof}
When $H_u = H$ for all vertices $u$ in $G$, Equality~\eqref{eq2} in the proof of Theorem~\ref{thm2} is equivalent to $[a,an]\cup[b,bn]$ being an interval for every pair $a,b$ bounding a non-dp interval, which is equivalent to $b\le an+1$.
\end{proof}

The next result is an immediate corollary of Theorem~\ref{thm2}.

\begin{cor}\label{cor1}
If $G$ is dp, with $|G| \ge 2$, then $G[\mathcal{H}]$ is dp for any set of graphs $\mathcal{H}$.
\end{cor}

Since any tree is dp, Corollary \ref{cor1} implies that the graph in Figure~\ref{fig:GFH} is dp.
The graph~$G[\mathcal{H}]$ being dp does not necessarily imply that $G$ is dp. This can be seen in Figure~\ref{fig:figure2} which shows the graph $C_5[\mathcal{H}]$, where~$C_5$ is the $5$-cycle and~$\mathcal{H}$ substitutes $K_2$ for one vertex and $K_1$ for all others. It is straightforward to verify that $C_5[\mathcal{H}]$ is dp, however $C_5$ is non-dp.

The next result follows easily from Lemma~\ref{L1} and Statement~\eqref{project} in the proof of Theorem~\ref{thm2}.
\begin{cor}\label{cor:subgraph}
 For a connected graph $G$ with $|G| \ge 2$ and an induced subgraph $M$ of $G[\mathcal{H}]$,
\begin{align*}
  M\le G[\mathcal{H}]  \text{ if and only if }
\begin{cases}
\pi(M)\le G &   \textnormal{ when }|\pi(M)|\ge 2, \\
\diam(M)\le 2 & \textnormal{ when }|\pi(M)|=1.
\end{cases}
\end{align*}
\end{cor}

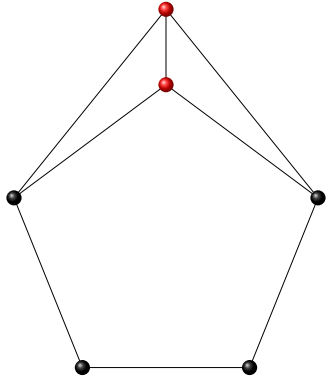
\begin{figure}[ht]
\centering
\begin{center}
    \begin{tikzpicture}	[thick, scale=0.5]
    

        \draw[line width=.01pt] (-2.2,-3.5) to  (2.2,-3.5);      
        \draw[line width=.01pt] (-4,1)to  (-2.2,-3.5);
        \draw[line width=.01pt] (4,1) to  (2.2,-3.5);

        \draw[line width=.01pt] (0,4) to  (4,1);
        \draw[line width=.01pt] (0,6) to  (4,1);
        \draw[line width=.01pt] (0,4) to  (-4,1);
        \draw[line width=.01pt] (0,6) to  (-4,1);
        \draw[line width=.01pt] (0,6) to  (0,4);
    
    	\shade[shading=ball, ball color=red] (0,6) circle (.2);
        \shade[shading=ball, ball color=red] (0,4) circle (.2);
	    \shade[shading=ball, ball color=black] (4,1) circle (.2);
	
		\shade[shading=ball, ball color=black] (-4,1) circle (.2);
        \shade[shading=ball, ball color=black] (-2.2,-3.5) circle (.2);
	    \shade[shading=ball, ball color=black] (2.2,-3.5) circle (.2);

	\end{tikzpicture}
\end{center}

\caption[]{The graph $C_5[\mathcal{H}]$, where~$\mathcal{H}$ substitutes $K_2$ for one vertex and $K_1$ for all others}
\label{fig:figure2}
\end{figure}

Recall that a graph $G$ is \emph{sequentially distance preserving}, which we denote sdp, if we can order the vertices $v_1,\ldots,v_n$ of $G$ such that the graph induced by $v_1,\ldots,v_k$ is isometric, for all~$1\le i\le n$. Corollary~\ref{cor:subgraph} implies the following result on sdp graphs:

\begin{cor}
If $G$ is sdp, with $|G| \ge 2$, then $G[\mathcal{H}]$ is sdp for any set of graphs $\mathcal{H}$.
\end{cor}

An illustrative example is shown in Figure~\ref{fig:figure3}. This 
figure depicts a graph $M$ formed of a social network of friendships between $44$ members of a community of international students, along with a modular decomposition of $M$ and a minimal quotient graph of $M$. By the results of Section~\ref{dp_graph}, to show that $M$ is distance preserving it is sufficient to show that the quotient graph is distance preserving. Note that the quotient graph does not contain any induced cycles of length greater than $4$, so the quotient graph is dp by Theorem 3.5 of \cite{SmithZahedi16}. Therefore, Corollary~\ref{cor1} implies that $M$ is distance preserving.

%
%
%
%
%
\begin{figure}
\centering
\begin{subfigure}[b]{0.8\textwidth}
\centering
\begin{tikzpicture}
\tikzset{VertexStyle/.style = {shape = circle,inner sep=0pt,minimum size = 0pt,draw=black}}
\tikzset{EdgeStyle/.style = {draw=black}}

\Vertex[L=\tiny\hbox{$1$},x=4cm,y=2.5cm]{v0}
\Vertex[L=\tiny\hbox{$2$},x=3.75cm,y=2cm]{v1}
\Vertex[L=\tiny\hbox{$3$},x=4.3cm,y=1.75cm]{v2}
\Vertex[L=\tiny\hbox{$4$},x=4.3cm,y=2.2cm]{v3}
\Vertex[L=\tiny\hbox{$5$},x=4cm,y=1.5cm]{v4}
\Vertex[L=\tiny\hbox{$6$},x=3.75cm,y=1.5cm]{v5}
\Vertex[L=\tiny\hbox{$7$},x=4.1cm,y=2cm]{v6}
\Vertex[L=\tiny\hbox{$8$},x=3cm,y=1cm]{v7}
\Vertex[L=\tiny\hbox{$9$},x=3.7cm,y=2.75cm]{v8}
\Vertex[L=\tiny\hbox{$10$},x=3.8991cm,y=4.2cm]{v9}
\Vertex[L=\tiny\hbox{$11$},x=3.1584cm,y=4.8949cm]{v10}
\Vertex[L=\tiny\hbox{$12$},x=3cm,y=4.4643cm]{v11}
\Vertex[L=\tiny\hbox{$13$},x=3.9551cm,y=3.7821cm]{v12}
\Vertex[L=\tiny\hbox{$14$},x=4.0928cm,y=3.4928cm]{v13}
\Vertex[L=\tiny\hbox{$15$},x=3.7923cm,y=4.6207cm]{v14}
\Vertex[L=\tiny\hbox{$16$},x=3.4907cm,y=4.8382cm]{v15}
\Vertex[L=\tiny\hbox{$17$},x=1.2698cm,y=4.9777cm]{v16}
\Vertex[L=\tiny\hbox{$18$},x=1.6566cm,y=4.9324cm]{v17}
\Vertex[L=\tiny\hbox{$19$},x=2cm,y=4cm]{v18}
\Vertex[L=\tiny\hbox{$20$},x=0.8cm,y=5.0cm]{v19}
\Vertex[L=\tiny\hbox{$21$},x=2.1cm,y=4.5cm]{v20}
\Vertex[L=\tiny\hbox{$22$},x=2.2cm,y=5cm]{v21}
\Vertex[L=\tiny\hbox{$23$},x=0.8612cm,y=4.4362cm]{v22}
\Vertex[L=\tiny\hbox{$24$},x=1.4891cm,y=3.6254cm]{v23}
\Vertex[L=\tiny\hbox{$25$},x=0.6cm,y=4cm]{v24}
\Vertex[L=\tiny\hbox{$26$},x=2cm,y=3.5cm]{v25}
\Vertex[L=\tiny\hbox{$27$},x=3.25cm,y=2.75cm]{v26}
\Vertex[L=\tiny\hbox{$28$},x=2.5784cm,y=2cm]{v27}
\Vertex[L=\tiny\hbox{$29$},x=2cm,y=1.75cm]{v28}
\Vertex[L=\tiny\hbox{$30$},x=3.042cm,y=1.6cm]{v29}
\Vertex[L=\tiny\hbox{$32$},x=1.7282cm,y=2.3cm]{v31}
\Vertex[L=\tiny\hbox{$33$},x=0.5cm,y=3.5cm]{v32}
\Vertex[L=\tiny\hbox{$34$},x=1.1cm,y=3.2cm]{v33}
\Vertex[L=\tiny\hbox{$35$},x=0cm,y=3.7932cm]{v34}
\Vertex[L=\tiny\hbox{$36$},x=0.9cm,y=4.1cm]{v35}
\Vertex[L=\tiny\hbox{$37$},x=1.2167cm,y=2cm]{v36}
\Vertex[L=\tiny\hbox{$38$},x=0.0cm,y=2.0433cm]{v37}
\Vertex[L=\tiny\hbox{$39$},x=0.947cm,y=2.5723cm]{v38}
\Vertex[L=\tiny\hbox{$40$},x=0.1cm,y=1.5215cm]{v39}
\Vertex[L=\tiny\hbox{$41$},x=0.5966cm,y=1.1022cm]{v40}
\Vertex[L=\tiny\hbox{$42$},x=0.0106cm,y=2.9757cm]{v41}
\Vertex[L=\tiny\hbox{$43$},x=0.1018cm,y=2.4125cm]{v42}
\Vertex[L=\tiny\hbox{$44$},x=0.4536cm,y=1.5675cm]{v43}
\Vertex[L=\tiny\hbox{$31$},x=3.1893cm,y=3.4494cm]{v44}
\Edge(v0)(v1)
\Edge(v1)(v2)
\Edge(v1)(v3)
\Edge(v1)(v4)
\Edge(v1)(v5)
\Edge(v1)(v6)
\Edge(v1)(v7)
\Edge(v1)(v8)
\Edge(v1)(v26)
\Edge(v1)(v27)
\Edge(v1)(v29)
\Edge(v1)(v44)
\Edge(v2)(v3)
\Edge(v2)(v4)
\Edge(v3)(v6)
\Edge(v4)(v5)
\Edge(v4)(v6)
\Edge(v5)(v6)
\Edge(v7)(v27)
\Edge(v8)(v44)
\Edge(v9)(v44)
\Edge(v10)(v44)
\Edge(v11)(v44)
\Edge(v12)(v44)
\Edge(v13)(v44)
\Edge(v14)(v44)
\Edge(v15)(v44)
\Edge(v16)(v19)
\Edge(v16)(v23)
\Edge(v17)(v23)
\Edge(v18)(v20)
\Edge(v18)(v23)
\Edge(v19)(v23)
\Edge(v20)(v21)
\Edge(v22)(v23)
\Edge(v22)(v34)
\Edge(v23)(v24)
\Edge(v23)(v25)
\Edge(v23)(v31)
\Edge(v23)(v32)
\Edge(v23)(v33)
\Edge(v23)(v34)
\Edge(v23)(v35)
\Edge(v23)(v36)
\Edge(v23)(v38)
\Edge(v24)(v38)
\Edge(v25)(v26)
\Edge(v25)(v27)
\Edge(v25)(v31)
\Edge(v25)(v38)
\Edge(v25)(v44)
\Edge(v26)(v27)
\Edge(v27)(v28)
\Edge(v27)(v29)
\Edge(v27)(v31)
\Edge(v27)(v44)
\Edge(v31)(v38)
\Edge(v32)(v38)
\Edge(v33)(v38)
\Edge(v34)(v38)
\Edge(v35)(v38)
\Edge(v36)(v38)
\Edge(v37)(v38)
\Edge(v38)(v39)
\Edge(v38)(v40)
\Edge(v38)(v41)
\Edge(v38)(v42)
\Edge(v38)(v43)
\end{tikzpicture}
\caption{A Graph $M$}   
\end{subfigure}
\begin{subfigure}[b]{0.59\textwidth}\centering
\begin{tikzpicture}
\tikzset{VertexStyle/.style = {shape = circle,inner sep=0pt,minimum size = 0pt,draw=black}}
\tikzset{EdgeStyle/.style = {draw=black}}

\Vertex[style={fill=green},L=\tiny\hbox{$1$},x=4cm,y=2.5cm]{v0}
\Vertex[L=\tiny\hbox{$2$},x=3.75cm,y=2cm]{v1}
\Vertex[style={fill=green},L=\tiny\hbox{$3$},x=4.3cm,y=1.75cm]{v2}
\Vertex[style={fill=green},L=\tiny\hbox{$4$},x=4.3cm,y=2.2cm]{v3}
\Vertex[style={fill=green},L=\tiny\hbox{$5$},x=4cm,y=1.5cm]{v4}
\Vertex[style={fill=green},L=\tiny\hbox{$6$},x=3.75cm,y=1.5cm]{v5}
\Vertex[style={fill=green},L=\tiny\hbox{$7$},x=4.1cm,y=2cm]{v6}
\Vertex[style={fill=pink},L=\tiny\hbox{$8$},x=3cm,y=1cm]{v7}
\Vertex[L=\tiny\hbox{$9$},x=3.7cm,y=2.75cm]{v8}
\Vertex[style={fill=red},L=\tiny\hbox{$10$},x=3.8991cm,y=4.2cm]{v9}
\Vertex[style={fill=red},L=\tiny\hbox{$11$},x=3.1584cm,y=4.8949cm]{v10}
\Vertex[style={fill=red},L=\tiny\hbox{$12$},x=3cm,y=4.4643cm]{v11}
\Vertex[style={fill=red},L=\tiny\hbox{$13$},x=3.9551cm,y=3.7821cm]{v12}
\Vertex[style={fill=red},L=\tiny\hbox{$14$},x=4.0928cm,y=3.4928cm]{v13}
\Vertex[style={fill=red},L=\tiny\hbox{$15$},x=3.7923cm,y=4.6207cm]{v14}
\Vertex[style={fill=red},L=\tiny\hbox{$16$},x=3.4907cm,y=4.8382cm]{v15}
\Vertex[style={fill=orange},L=\tiny\hbox{$17$},x=1.2698cm,y=4.9777cm]{v16}
\Vertex[style={fill=orange},L=\tiny\hbox{$18$},x=1.6566cm,y=4.9324cm]{v17}
\Vertex[L=\tiny\hbox{$19$},x=2cm,y=4cm]{v18}
\Vertex[style={fill=orange},L=\tiny\hbox{$20$},x=0.8cm,y=5.0cm]{v19}
\Vertex[L=\tiny\hbox{$21$},x=2.1cm,y=4.5cm]{v20}
\Vertex[L=\tiny\hbox{$22$},x=2.2cm,y=5cm]{v21}
\Vertex[L=\tiny\hbox{$23$},x=0.8612cm,y=4.4362cm]{v22}
\Vertex[L=\tiny\hbox{$24$},x=1.4891cm,y=3.6254cm]{v23}
\Vertex[style={fill=brown},L=\tiny\hbox{$25$},x=0.6cm,y=4cm]{v24}
\Vertex[L=\tiny\hbox{$26$},x=2cm,y=3.5cm]{v25}
\Vertex[L=\tiny\hbox{$27$},x=3.25cm,y=2.75cm]{v26}
\Vertex[L=\tiny\hbox{$28$},x=2.5784cm,y=2cm]{v27}
\Vertex[L=\tiny\hbox{$29$},x=2cm,y=1.75cm]{v28}
\Vertex[style={fill=pink},L=\tiny\hbox{$30$},x=3.042cm,y=1.6cm]{v29}
\Vertex[L=\tiny\hbox{$32$},x=1.7282cm,y=2.3cm]{v31}
\Vertex[style={fill=brown},L=\tiny\hbox{$33$},x=0.5cm,y=3.5cm]{v32}
\Vertex[style={fill=brown},L=\tiny\hbox{$34$},x=1.1cm,y=3.2cm]{v33}
\Vertex[L=\tiny\hbox{$35$},x=0cm,y=3.7932cm]{v34}
\Vertex[style={fill=brown},L=\tiny\hbox{$36$},x=0.9cm,y=4.1cm]{v35}
\Vertex[style={fill=brown},L=\tiny\hbox{$37$},x=1.2167cm,y=2cm]{v36}
\Vertex[style={fill=Blue},L=\tiny\hbox{$38$},x=0.0cm,y=2.0433cm]{v37}
\Vertex[L=\tiny\hbox{$39$},x=0.947cm,y=2.5723cm]{v38}
\Vertex[style={fill=Blue},L=\tiny\hbox{$40$},x=0.1cm,y=1.5215cm]{v39}
\Vertex[style={fill=Blue},L=\tiny\hbox{$41$},x=0.5966cm,y=1.1022cm]{v40}
\Vertex[style={fill=Blue},L=\tiny\hbox{$42$},x=0.0106cm,y=2.9757cm]{v41}
\Vertex[style={fill=Blue},L=\tiny\hbox{$43$},x=0.1018cm,y=2.4125cm]{v42}
\Vertex[style={fill=Blue},L=\tiny\hbox{$44$},x=0.4536cm,y=1.5675cm]{v43}
\Vertex[L=\tiny\hbox{$31$},x=3.1893cm,y=3.4494cm]{v44}
\Edge(v0)(v1)
\Edge(v1)(v2)
\Edge(v1)(v3)
\Edge(v1)(v4)
\Edge(v1)(v5)
\Edge(v1)(v6)
\Edge(v1)(v7)
\Edge(v1)(v8)
\Edge(v1)(v26)
\Edge(v1)(v27)
\Edge(v1)(v29)
\Edge(v1)(v44)
\Edge(v2)(v3)
\Edge(v2)(v4)
\Edge(v3)(v6)
\Edge(v4)(v5)
\Edge(v4)(v6)
\Edge(v5)(v6)
\Edge(v7)(v27)
\Edge(v8)(v44)
\Edge(v9)(v44)
\Edge(v10)(v44)
\Edge(v11)(v44)
\Edge(v12)(v44)
\Edge(v13)(v44)
\Edge(v14)(v44)
\Edge(v15)(v44)
\Edge(v16)(v19)
\Edge(v16)(v23)
\Edge(v17)(v23)
\Edge(v18)(v20)
\Edge(v18)(v23)
\Edge(v19)(v23)
\Edge(v20)(v21)
\Edge(v22)(v23)
\Edge(v22)(v34)
\Edge(v23)(v24)
\Edge(v23)(v25)
\Edge(v23)(v31)
\Edge(v23)(v32)
\Edge(v23)(v33)
\Edge(v23)(v34)
\Edge(v23)(v35)
\Edge(v23)(v36)
\Edge(v23)(v38)
\Edge(v24)(v38)
\Edge(v25)(v26)
\Edge(v25)(v27)
\Edge(v25)(v31)
\Edge(v25)(v38)
\Edge(v25)(v44)
\Edge(v26)(v27)
\Edge(v27)(v28)
\Edge(v27)(v29)
\Edge(v27)(v31)
\Edge(v27)(v44)
\Edge(v31)(v38)
\Edge(v32)(v38)
\Edge(v33)(v38)
\Edge(v34)(v38)
\Edge(v35)(v38)
\Edge(v36)(v38)
\Edge(v37)(v38)
\Edge(v38)(v39)
\Edge(v38)(v40)
\Edge(v38)(v41)
\Edge(v38)(v42)
\Edge(v38)(v43)
\end{tikzpicture}
\caption{The Maximal Modular Decomposition of $M$}  
\end{subfigure}
\begin{subfigure}[b]{0.4\textwidth}\centering
\begin{tikzpicture}
\tikzset{VertexStyle/.style = {shape = circle,inner sep=0pt,minimum size = 0pt,draw=black}}
\tikzset{EdgeStyle/.style = {draw=black}}

\Vertex[style={fill=green},L=\tiny\hbox{\textcolor{green}{1}},x=4cm,y=2.5cm]{v0}
\Vertex[L=\tiny\hbox{$2$},x=3.75cm,y=2cm]{v1}
\Vertex[style={fill=pink},L=\tiny\hbox{\textcolor{pink}{1}},x=3cm,y=1cm]{v7}
\Vertex[L=\tiny\hbox{$9$},x=3.7cm,y=2.75cm]{v8}
\Vertex[style={fill=red},L=\tiny\hbox{\textcolor{red}{1}},x=3.8991cm,y=4.2cm]{v9}
\Vertex[style={fill=orange},L=\tiny\hbox{\textcolor{orange}{1}},x=1.2698cm,y=4.9777cm]{v16}
\Vertex[L=\tiny\hbox{$19$},x=2cm,y=4cm]{v18}
\Vertex[L=\tiny\hbox{$21$},x=2.1cm,y=4.5cm]{v20}
\Vertex[L=\tiny\hbox{$22$},x=2.2cm,y=5cm]{v21}
\Vertex[L=\tiny\hbox{$23$},x=0.8612cm,y=4.4362cm]{v22}
\Vertex[L=\tiny\hbox{$24$},x=1.4891cm,y=3.6254cm]{v23}
\Vertex[style={fill=brown},L=\tiny\hbox{\textcolor{brown}{1}},x=0.6cm,y=4cm]{v24}
\Vertex[L=\tiny\hbox{$26$},x=2cm,y=3.5cm]{v25}
\Vertex[L=\tiny\hbox{$27$},x=3.25cm,y=2.75cm]{v26}
\Vertex[L=\tiny\hbox{$28$},x=2.5784cm,y=2cm]{v27}
\Vertex[L=\tiny\hbox{$29$},x=2cm,y=1.75cm]{v28}
\Vertex[L=\tiny\hbox{$32$},x=1.7282cm,y=2.3cm]{v31}
\Vertex[L=\tiny\hbox{$35$},x=0cm,y=3.7932cm]{v34}
\Vertex[style={fill=Blue},L=\tiny\hbox{\textcolor{Blue}{1}},x=0.0cm,y=2.0433cm]{v37}
\Vertex[L=\tiny\hbox{$39$},x=0.947cm,y=2.5723cm]{v38}
\Vertex[L=\tiny\hbox{$45$},x=3.1893cm,y=3.4494cm]{v44}
\Edge(v0)(v1)
\Edge(v1)(v7)
\Edge(v1)(v8)
\Edge(v1)(v26)
\Edge(v1)(v27)
\Edge(v1)(v44)
\Edge(v7)(v27)
\Edge(v8)(v44)
\Edge(v9)(v44)
\Edge(v16)(v23)
\Edge(v18)(v20)
\Edge(v18)(v23)
\Edge(v20)(v21)
\Edge(v22)(v23)
\Edge(v22)(v34)
\Edge(v23)(v24)
\Edge(v23)(v25)
\Edge(v23)(v31)
\Edge(v23)(v34)
\Edge(v23)(v38)
\Edge(v24)(v38)
\Edge(v25)(v26)
\Edge(v25)(v27)
\Edge(v25)(v31)
\Edge(v25)(v38)
\Edge(v25)(v44)
\Edge(v26)(v27)
\Edge(v27)(v28)
\Edge(v27)(v31)
\Edge(v27)(v44)
\Edge(v31)(v38)
\Edge(v34)(v38)
\Edge(v37)(v38)
\end{tikzpicture}
\caption{The Minimal Quotient Graph of $M$}
\end{subfigure}
\caption{}
\label{fig:figure3}
\end{figure}

\section{Cartesian Product Graphs}\label{section:F}
In~\cite{khalifeh2015distance} the behaviour of the distance preserving property is investigated with respect to the Cartesian product of graphs. Recall that the Cartesian product of two graphs $G$ and $H$, denoted $G\ \Box\ H$, has vertex set $V(G)\times V(H)$ and two vertices $(g,h)$ and $(g',h')$ are adjacent precisely if $g=g'$ and $hh'\in E(H)$ or $h=h'$ and $gg'\in E(G)$.

 It was shown in~\cite{khalifeh2015distance} that $G$ and $H$ are sequentially distance preserving if and only if~$G\ \Box\ H$ is sequentially distance preserving. Furthermore, it was conjectured that if $G$ and~$H$ are distance preserving, then so is $G\ \Box\ H$. We prove a result somewhat weaker than the conjecture. To this end we need the following notation and lemma,
$$ \DP'(G):=\bigl\{ A\subseteq V(G)\,\big|\, G-A \le G\bigl\}\ \ \text{  and }\ \ \ddp'(G):=\bigl\{ |A|\,\big|\, A \in \DP'(G)\bigl\}.$$

\begin{lem}\label{lemma:2}~\cite[Lemma 4.3]{khalifeh2015distance}
Given nonempty subsets $A$ and $B$ of the vertex sets of graphs $G$ and $H$, respectively, then 
$A\times B\in \DP'(G\ \Box\  H)$ if and only if $A\in \DP'(G)$ and $B\in \DP'(H)$.
\end{lem}

Now we have all we need to prove the main result for this section. Which follows a similar argument to that of Theorem 4.4 in \cite{khalifeh2015distance}.
 \begin{prop}\label{proposition:2}
 If $G$ is sequentially distance preserving and $H$ is distance preserving, then $G\,\Box\, H$ is distance preserving.
\end{prop}
\begin{proof}
Since $G$ is sdp there is an ordering $v_1,\ldots,v_{|G|}$ of~$V(G)$ such that $\{v_i\}_{i=1}^s\in\DP'(G)$, for every $1\leq s\leq |G|$.
Moreover, since $H$ is distance preserving there is a set $A_j\in\DP'(H)$ with $|A_j|=j$, for every $1\leq j\leq |H|$. By Lemma \ref{lemma:2} we know that ${v_1}\times A_j\in \DP'(G\ \Box\  H)$, for all $1\le j\le |H|$. Furthermore, Lemma \ref{lemma:2} implies that ${v_2}\times A_j\in \DP'((G-\{v\})\ \Box\  H)$, for all $1\le j\le |H|$, so by the transitivity of the isometric property we get $(v_1\times H)\cup(v_2\times A_j)\in \DP'(G\ \Box\  H)$. Applying this argument inductively we get $(\{v_i\}_{i=1}^{s-1}\times H)\cup(v_s\times A_j)\in \DP'(G\ \Box\  H)$, for all $1\leq s\leq |G|$ and $1\leq j\leq |H|$. Therefore, $[0,|G|\times |H|]\subseteq \ddp'(G\ \Box\  H)$, so $G\ \Box\  H$ is dp.
\end{proof}

\end{document}